\documentclass[conference]{IEEEtran}
\usepackage{mathtools}
\usepackage{amsmath}
\usepackage{amssymb}
\usepackage{amsthm}
\usepackage{dsfont} 
\usepackage{xcolor}
\usepackage{cite}
\usepackage{hyperref}
\usepackage{lipsum}
\usepackage{scalerel}

\newcounter{notecounter}

\newenvironment{noted}[1][]
{\unskip\space\ignorespaces\refstepcounter{notecounter}}
{\unskip\space\ignorespacesafterend}

\newcommand{\note}[1]{\text{\begin{noted}$\stackrel{\text{(\alph{notecounter})}}{#1}$\end{noted}}}
\makeatletter
\newcommand{\makeAlph}[1]{\@alph{#1}}
\makeatother
\newcommand{\refnote}[1]{(\makeAlph{\numexpr\value{notecounter}-#1\relax})}

\newcommand{\set}[1]{\mathcal{#1}}
\newcommand{\sset}[1]{\mathbb{#1}}

\DeclarePairedDelimiter\brackets{[}{]}
\DeclarePairedDelimiter\parenthesis{(}{)}
\DeclarePairedDelimiter\braces{\{}{\}}
\DeclarePairedDelimiter\bars{\lvert}{\rvert}
\DeclarePairedDelimiter\floor{\lfloor}{\rfloor}
\DeclarePairedDelimiter\ceil{\lceil}{\rceil}

\newcommand{\expectation}[2][]{\sset{E}\brackets[#1]{#2}}
\newcommand{\probability}[2][]{\sset{P}\brackets[#1]{#2}}
\newcommand{\variance}[2][]{\mathrm{Var}\parenthesis[#1]{#2}}

\newcommand{\indicator}[2][]{\mathds{1}\braces[#1]{#2}}
\newcommand{\cardinality}[2][]{\bars[#1]{#2}}

\newcommand{\abs}[2][]{\bars[#1]{#2}}
\DeclarePairedDelimiterXPP\sqrd[1]{}{(}{)}{^2}{#1}

\newcommand{\variationaldistance}[3][]{\mathbb{V}\parenthesis[#1]{#2, #3}}
\newcommand{\typicalset}[2][]{\mathcal{T}_\epsilon^n\parenthesis[#1]{#2}}

\DeclareMathOperator*{\concat}{\scalerel*{\Vert}{\sum}}

\newtheorem{theorem}{Theorem}

\theoremstyle{definition}
\newtheorem{definition}{Definition}

\newtheorem{proposition}{Proposition}

\begin{document}
\title{Secret Sharing with Additive Access Structures from Correlated Random Variables}
\author{\IEEEauthorblockN{David Miller, R\'{e}mi A. Chou}
\IEEEauthorblockA{\textit{Department of Computer Science and Engineering} \\
\textit{University of Texas at Arlington}\\
Arlington, TX, USA \\
david.miller2@mavs.uta.edu, remi.chou@uta.edu\thanks{This work was supported in part by NSF grant CCF-2047913 and CCF-2401373.  }} 
}
\IEEEoverridecommandlockouts

\maketitle

\begin{abstract}
    We generalize secret-sharing models that rely on correlated randomness and public communication, originally designed for a fixed access structure, to support a sequence of dynamic access structures, which we term an Additive Access Structure. Specifically, the access structure  is allowed to monotonically grow by having any subset of participants added to it at a given time step, and the dealer only learns of these changes to the access structure on the time step that they occur. For this model, we prove the existence of a secret sharing strategy that achieves the same secret rate at each time step as the best known strategy for the fixed access structure version of this model. We also prove that there exists a strategy that is capacity-achieving at any time step where the access structure is a threshold access structure.
\end{abstract}

\allowdisplaybreaks
\section{Introduction}
Secret Sharing is the problem of distributing private pieces of information called shares to a set of participants such that the combined shares of an authorized subset of participants gives enough information to recover a secret, but the combined shares of an unauthorized subset of participants reveals no information about this secret. Since secret sharing was introduced in \cite{Shamir_Secret_Sharing}, it has seen a wide variety of applications such as secure multiparty computation \cite{Multiparty_unconditionally_secure_protocols}, \mbox{Byzantine agreement \cite{byzantine_generals}}, threshold cryptography \cite{Shared_generation_of_authenticators_and_signatures}, and generalized oblivious \mbox{transfer \cite{Alternative_Protocols_for_Generalized_Oblivious_Transfer}}. Most of the work on secret sharing has focused on models where shares are distributed using private channels, e.g., \cite{Secret_Sharing_Survey}. In these models, efficiency is measured by minimizing the share size relative to the secret size. There is also existing work on secret sharing that replaces private channels with correlated randomness combined with a public channel \cite{zou2015information, 
Distributed_Secret_Sharing, 
Capacity_of_a_shared_secret_key,rana2021information}.  In this line of work, efficiency is quantified by the secret rate, defined as the size of the shared secret normalized by the number of source observations. However, in both of these models, the strategies assume the access structure remains fixed. If the access structure were to change, all of the shares would have to be recreated, potentially leading to inefficiency. Instead, if shares were able to be reused as the access structure changed, this inefficiency could be mitigated. For instance, this idea is pursued by \cite{How_to_share_a_secret_infinitely} which introduced the Evolving Access Structure (EAS): a model where shares are distributed to participants as they join the system and the access structure can grow monotonically by adding authorized groups formed with the participant that most recently joined. Further work on the EAS includes \cite{Evolving_Secret_Sharing, Evolving_Secret_Sharing_Made_Short, Secret_Sharing_on_Evolving_Multi-level_Access_Structure}. 

In this paper, instead of considering a model where participants join one by one over time, and have shares communicated to them over a private channel, we consider a model where all  the participants are part of the model initially, shares are produced by a correlated source of randomness, and the dealer communicates to the participants using a public channel. We call our model an Additive Access Structure (AAS) to distinguish it from the EAS since it differs in two major ways. First, there is no notion of participants joining, and as a result, there is no restriction on which subsets of participants can be added to the access structure at a given time step. Second, since all the shares are created simultaneously from a correlated source of randomness, such as channel gain measurements obtained in a wireless communication \mbox{network 
\cite{Channel_Identification_Secret_Sharing_Using_Reciprocity_in_Ultrawideband_Channels, 
Automatic_Secret_Keys_From_Reciprocal_MIMO_Wireless_Channels_Measurement_and_Analysis, 
Information-Theoretically_Secret_Key_Generation_for_Fading_Wireless_Channels, 
Experimental_aspects_of_secret_key_generation_in_indoor_wireless_environments}}, public messages are the mechanism used to allow newly authorized groups to gain enough information to recover the shared secret.

More specifically, our model for secret sharing can be described as follows. Initially, the dealer and each participant observe independent and identically distributed realizations of a correlated source of randomness and the access structure will begin with no authorized groups. Whenever one or more authorized groups are added to the access structure at a given time step, the dealer will learn the new access structure, determine a secret to be shared, and send message to every participant over a public, noiseless, authenticated channel. Since this is an iterative process, we will define an AAS as the sequence of access structures produced by adding authorized groups, and have each access structure in the AAS indexed by the time step it is produced. 

For example, consider the case where there are four participants $\set{L} = \{1,2,3,4\}$. At the first time step, it might be decided that the group $\{1,2,4\}$ should become authorized. This will make the access structure at $t=1$ the pair $(\sset{A}_1, \sset{U}_1)$ where $\sset{A}_1 = \{\{1,2,4\}, \{1,2,3,4\}\}$ and $\sset{U}_1 = 2^{\set{L}} \backslash \sset{A}_1$. The dealer, with the knowledge of $(\sset{A}_1, \sset{U}_1)$, will then determine the secret $S_1$ to be shared, and the message $M_1$ to be sent. 
Then, at $t=2$, it might be decided the group $\{2,3\}$ should become authorized giving $\sset{A}_2 = \sset{A}_1 \cup \{\{2,3\}, \{1,2,3\}, \{2,3,4\}, \{1,2,3,4\}\}$ and $\sset{U}_2 = 2^{\set{L}} \backslash \sset{A}_2$. Now, the dealer having learned about $(\sset{A}_2, \sset{U}_2)$ and knowing the message $M_1$ that was previously sent, will determine a new secret $S_2$  to be shared, and a new message $M_2$ to be sent. 
This process could continue for as long as there are groups which are not yet authorized, but importantly cannot go on forever since $2^{\set{L}}$ is finite. If no more groups were added to the access structure after $t=2$, the AAS would be the sequence $((\sset{A}_1, \sset{U}_1),(\sset{A}_2, \sset{U}_2))$.

A secret rate tuple for the dealer's strategy will be considered achievable if at each time step every authorized group, using their observations and all public messages, can determine the secret for that time step with near certainty and if every unauthorized group, using their observations and all public messages, learns a negligible amount of information about the secret.

We then prove the surprising result that there exists a secret sharing strategy for an AAS where the secret rate achievable at a given time step is independent of the access structures at all other time steps, and that this secret rate matches the best known achievable secret rate for a fixed access \mbox{structure \cite{Distributed_Secret_Sharing}}. Additionally, we prove that at any time step where the access structure  forms a threshold access structure, the secret rate at that time step is equal to the secret sharing capacity of that threshold access structure. For the threshold AAS case, i.e., when at each time step one has a threshold access structure, it means that a dealer who is omniscient and knows the full AAS before sending any public messages has no advantage over a dealer who has to commit to a message at the current time step before learning which subsets of participants will be added to the access structure in the next time step.
In our achievability proof, we introduce a novel adaptation of random binning in which the number of binnings is quantized. This modification enables us to establish the existence of a single secret sharing strategy that is fully specified a priori and does not rely on any knowledge of how the AAS evolves over~time.

The remainder of this paper is organized as follows. In Section \ref{Problem Statement}, we define the setting of the problem. In \mbox{Section \ref{Results}}, we present the main results. In Section \ref{Proof of Theorem 1}, we prove the achievability result. Finally, in Section \ref{Conclusion}, we give some concluding remarks.

\section{Problem Statement} \label{Problem Statement}
Notation: $2^{\set{A}}$ is the powerset of a set $\set{A}$. 
$\cardinality{\set{A}}$ is the cardinality of a set $\set{A}$.
$\set{A} \backslash \set{B}$ is the set difference between sets $\set{A}$ and $\set{B}$.
For real numbers $a$ and $b$, $[a, b] \triangleq \{x \in \sset{N} : a \leq x \leq b\}$.
$\indicator{\cdot}$ is the indicator function.
$\expectation{A}$ is the expectation of a random variable $A$. $p_{X}$ is the probability distribution for a random variable $X$, which may sometimes be written as $p$ if the random variable is understood from context.  
$\variationaldistance{p_X}{\dot{p}_X} \triangleq \sum_{x \in \set{X}} \abs{p_X(x) - \dot{p}_X(x)}$ is the variational distance between two probability distributions defined over the same alphabet $\set{X}$.
The components of a tuple $x^n$ of length $n \in \sset{N}$ are denoted with subscripts, i.e., $x^n \triangleq (x_1, \ldots, x_n)$.

\begin{definition} \label{AAS Definition}
	For a set of participants $\set{L}$ indexed by $\set{L} \triangleq [1, L]$, an Additive Access Structure (AAS) is a sequence of Access Structures $(\sset{A}_t, \sset{U}_t)_{t \in \set{T}}$ over a set of time steps $\set{T} \triangleq [1, t_{max}]$, $\sset{A}_t \subseteq 2^{\set{L}}$ being the sets of participants considered authorized and $\sset{U}_t \subseteq 2^\set{L} \backslash \sset{A}_t$ being the sets of participants considered unauthorized at time $t$, that satisfies the following properties:
	\begin{enumerate}
        \item \label{AAS Definition: 1}
		$\forall t \in \set{T}$, $\sset{A}_{t-1} \subset \sset{A}_{t} $, where $\sset{A}_0 \triangleq \emptyset$. In other words, at least one group is newly authorized at every time step, and once a group is authorized they never become unauthorized.
        \item \label{AAS Definition: 2}
		$\forall t \in \set{T}$, $\sset{U}_{t-1} \supseteq \sset{U}_{t} $, where $\sset{U}_0 \triangleq 2^{\set{L}}$. In other words, once a group is no longer unauthorized, it can never become unauthorized again. 
        \item \label{AAS Definition: 3}
		$\forall t \in \set{T}, \forall \set{B} \subseteq \set{L}$, if $\exists \set{A} \in \sset{A}_t$ such that $\set{A} \subseteq \set{B}$, then $\set{B} \in \sset{A}_t$. In other words, the access structure is monotone.
	\end{enumerate}
\end{definition} 
\begin{definition}
    Let $\set{Y}$ be a finite alphabet and $\set{X}_{\set{L}}$ be the Cartesian Product of $L$ finite alphabets $(\set{X}_{l})_{l \in \set{L}}$. 
    A $\left((2^{nR_t}, \sset{A}_t, \sset{U}_t)_{t \in \set{T}}, n\right)$ secret sharing strategy for an AAS with discrete memoryless source $(\set{Y} \times \set{X}_{\set{L}}, p_{YX_{\set{L}}}) $ consists of
    \begin{itemize}
        \item A key alphabet $\set{S}_t \triangleq [1, 2^{nR_t}], ~ \forall t \in \set{T}$;
        \item A message alphabet $\set{M}_t \triangleq [1, \ceil{2^{nR^\prime_t}}], ~ \forall t \in \set{T}$;
        \item An encoding function $e_t: \set{Y}^n \times 2^{\set{L}} \times 2^{\set{L}} \times \set{S}^{t-1} \times \set{M}^{t-1} \to \set{S}_t, ~ \forall t \in \set{T}$;
        \item An encoding function $f_t: \set{Y}^n \times 2^{\set{L}} \times 2^{\set{L}} \times \set{S}^{t} \times \set{M}^{t-1} \to \set{M}_t, ~ \forall t \in \set{T}$;
        \item A decoding function $d_t(\set{A}): \set{X}_{\set{A}}^n \times \set{M}^t \to \set{S}_t \cup \{error\}, ~ \forall \mathcal{A} \in \mathbb{A}_t$, $\forall t \in \set{T}$;
    \end{itemize}
    and operates as follows: \\
	Initially, the dealer observes $Y^n$ and each participant $l \in \set{L}$ observes $X_l^n$. Then, at each time $t \in \set{T}$ the following happens:
    \begin{enumerate}
        \item The dealer learns $\sset{A}_t$ and $\sset{U}_t$.
        \item The dealer computes a secret $S_t = e_t(Y^n, \sset{A}_t, \sset{U}_t, S^{t-1}, M^{t-1})$. 
        \item The dealer sends, over a noiseless public authenticated channel, a message $M_t = f_t(Y^n, \sset{A}_t, \sset{U}_t, S^{t}, M^{t-1})$ to every participant in $\set{L}$.
        \item Any subset of participants $\set{A} \in \sset{A}_t$ can compute 
        $\hat{S}_t = d_t(X_{\set{A}}^n, M^t)$.
    \end{enumerate}
\end{definition}

\begin{definition} \label{Achievability Definition}
	A secret rate tuple $(R_t)_{t \in \set{T}}$ is achievable if there exists a sequence of $\left((2^{nR_t}, \sset{A}_t, \sset{U}_t)_{t \in \set{T}}, n\right)$ secret sharing strategies such that for any $t \in \set{T}$: 
	\begin{align}
        \label{Achievability Definition: Reliability}
		\lim\limits_{n \to \infty} \max_{\set{A} \in \sset{A}_t} \probability{\hat{S}_t \neq S_t} &= 0, \tag{Reliability} \\
        \label{Achievability Definition: Secrecy}
		\lim\limits_{n \to \infty} \max_{\set{U} \in \sset{U}_t} I (S_t;M^t,X_{\set{U}}^n) &= 0, \tag{Secrecy} \\
        \label{Achievability Definition: Secret Uniformity}
		\lim\limits_{n \to \infty} \log\cardinality{\set{S}_t} - H(S_t) &= 0.
        \tag{Secret Uniformity}
	\end{align}
\end{definition}

\section{Results} \label{Results}
We begin with two theorems which give lower and upper bounds respectively on achievable secret  rates for an AAS.
\begin{theorem}[Achievability] \label{Achievability Result}
    The secret rate tuple $(R_t)_{t \in \set{T}}$ is achievable for an AAS, where
    \begin{equation*}
        R_t = \min_{\set{U} \in \sset{U}_t} H(Y|X_{\set{U}}) 
        - \max_{\set{A} \in \sset{A}_t} H(Y|X_{\set{A}}), 
        ~\forall t \in \set{T}.
    \end{equation*}
\end{theorem}
\begin{proof}
    See Section \ref{Proof of Theorem 1}. 
\end{proof}

One consequence of Theorem \ref{Achievability Result} is that the achievable rate $R_t$ at time $t$ only depends on the distribution $p_{YX_{\set{L}}}$ and the access structure $(\sset{A}_t, \sset{U}_t)$. This means that the achievable secret rate for a particular access structure at a given time step does not depend on how the AAS evolves. 
\begin{theorem}[Converse]
    \label{Converse Result}
    If a rate tuple $(R_t)_{t \in \set{T}}$ is achievable for an AAS, then it satisfies the constraints
    \begin{equation*}
        R_t \leq \min_{\set{U} \in \sset{U}_t} \min_{\set{A} \in \sset{A}_t} I(Y;X_{\set{A}}|X_{\set{U}}), 
        ~\forall t \in \set{T}.
    \end{equation*}
\end{theorem}
\begin{proof}
    For any fixed time $t$, we treat $(\sset{A}_t, \sset{U}_t)$ as a static access structure, in which case the converse result of \cite{Distributed_Secret_Sharing} applies.
\end{proof}

We now provide an example where Theorems \ref{Achievability Result} and \ref{Converse Result} match at any time $t$, for threshold AAS, defined as follows.
\begin{definition} \label{Threshold AAS}
A threshold AAS (TAAS) is an AAS, parameterized by functions $f_1: \set{T} \to [0, L]$ and $f_2: \set{T} \to [0, L]$, such that $(\sset{A}_t, \sset{U}_t)$ is a threshold access structure for every time step $t$, i.e., $\sset{A}_t = \{\set{A} \subseteq \set{L} : \cardinality{\set{A}} \geq f_1(t)\}$ and $\sset{U}_t = \{\set{U} \subseteq \set{L} : \cardinality{\set{U}} \leq f_2(t)\}$, where $f_1$ and $f_2$ represent the authorization and unauthorized thresholds, respectively, at each time $t$. 
\end{definition}

\begin{theorem}[Capacity Result]
\label{proposition:2}
Consider the case where $p_{YX_{\set{L}}}$ is defined such that $Y = X_{\set{L}}$ and $\forall l \in \set{L}, p_{X_l}(x_l) = \frac{1}{\cardinality{\set{X}}}$, which models the dealer sharing a private key of size $\cardinality{\set{X}}$ with each participant and where these private keys are drawn independently from a uniform distribution. For a TAAS parameterized by $f_1$ and $f_2$,  the secret  capacity at time $t$ is $R_t = H(X)\left(f_1(t)-f_2(t)\right)$, for all $t \in \set{T}$. 
\end{theorem}
\begin{proof}
See the appendix.
\end{proof}
\section{Proof of Theorem \ref{Achievability Result}} 
\label{Proof of Theorem 1}
Notation: $\concat_{i=k}^j x_i$ is the concatenation of strings $x_k, \ldots, x_j$.
$\delta(\epsilon)$ is a function of $\epsilon$ that satisfies $\lim_{\epsilon \to 0} \delta(\epsilon) = 0$. 
$\delta_{\epsilon}(n)$ is a function of $\epsilon$ and $n$ that satisfies $\lim_{n \to \infty} \delta_{\epsilon}(n) = 0$. 
$\typicalset{X} \triangleq \big\{ x^n \in \set{X}^n : \cardinality[\Big]{\dfrac{1}{n}\sum_{i=1}^n \indicator{x_i = x} - p_{X}(x)} \leq \epsilon p_{X}(x), \forall x \in \set{X} \big\}$ is the set of $\epsilon$-typical sequences $x^n$ associated with $p_{X}$, e.g.,  \cite{Physical_Layer_Security}. 
$\typicalset{XY|x^n} \triangleq \big\{ y^n \in \set{Y}^n : (x^n, y^n) \in \typicalset{XY} \big\}$ is the set of conditionally $\epsilon$-typical sequences $y^n$ associated with $p_{XY}$ with respect to $x^n \in \typicalset{X}$, e.g.,  \cite{Physical_Layer_Security}.
\subsection{Secret Sharing Strategy}
\label{Secret Sharing Strategy}
For $\epsilon > 0$, for each $y^n \in \set{Y}^n$, 
draw uniformly at random 
$b$ indices from the sets $\set{M}_\epsilon \triangleq [1, \ceil{2^{n\epsilon}}]$ and $\set{S}_\epsilon \triangleq [1, 2^{n\epsilon}]$ where
	\[b \triangleq 
		\ceil[\bigg]{ 
			\frac
				{\log \cardinality{ \set{Y} } }
				{\epsilon} 
		}.
	\]
Let these random binnings define the functions
\begin{align*}
	g_i: \set{Y}^n &\rightarrow \set{M}_\epsilon &\forall i \in [1, b], \\
	h_j: \set{Y}^n &\rightarrow \set{S}_\epsilon &\forall j \in [1, b].
\end{align*}

For each time $t$, $k_t: \sset{A}_t \times \sset{U}_t \to [0, b]$ and $\sigma_t: \sset{A}_t \times \sset{U}_t \to [0, b]$ will represent the number of quantized pieces $M_\epsilon$ and $S_\epsilon$ that are used to construct message $M_t$ and secret $S_t$ respectively. For  convenience, we  write $k_t$ and $\sigma_t$ instead of $k_t(\sset{A}_t, \sset{U}_t)$ and $\sigma_t(\sset{A}_t, \sset{U}_t)$ when the access structure is known from context. We define $k_0 \triangleq 0$ and $\sigma_0 \triangleq 0$. 
The encoding and decoding at time $t\in \mathcal{T}$ is as follows.

Encoding at the dealer:
Given $y^n$, $k_{t-1}$ and $\sigma_{t-1}$, choose values $k_t \geq k_{t-1}$ and $\sigma_t$ and output
\begin{align}
    \label{total message}
	m_t &\triangleq 
    \begin{dcases}
        \concat_{i = k_{t-1}+1}^{k_t} g_i(y^n), &k_t > k_{t-1} \\
        \emptyset, &k_t = k_{t-1} 
    \end{dcases}
    \\
    \label{total secret}
	s_t &\triangleq \concat_{j = 1}^{\sigma_t} h_j(y^n),
\end{align}
where $\emptyset$ is the empty string.

Decoding for a set of participants $\set{A} \in \sset{A}_t$: 
Given $m^t$, $k_t$, $\sigma_t$, and $x_{\set{A}}^n$, output $\hat{y}^n(\set{A})$ if it is the unique sequence such that 
	$(x_{\set{A}}^n,\hat{y}^n(\set{A})) \in \set{T}_\epsilon^n(X_{\set{A}} Y)$
and 
	$\concat_{i = 1}^{k_t} g_i(\hat{y}^n(\set{A})) = m^t$,
otherwise output $error$. 

Next, we prove that there exist binnings $g^{b}$ and $h^{b}$ that can be chosen without prior knowledge of the AAS and a method for how the dealer chooses $k_t$ and $\sigma_t$ at each time step $t$, that achieves the rates given in Theorem \ref{Achievability Result}. 

\subsection{Proof of Reliability} \label{Proof of Reliability}
\setcounter{notecounter}{0}
The event of an error in decoding is $\{\hat{Y}^n(\set{A}) \neq Y^n\} = \set{E}_0 \cup \set{E}_1$ where
\begin{align*}
    \set{E}_0 &\triangleq \{(X_{\set{A}}^n,Y^n) \notin \set{T}_\epsilon^n(X_{\set{A}} Y)\} \\
	\set{E}_1 &\triangleq \{\exists \hat{y}^n \neq Y^n, 
					(X_{\set{A}}^n,\hat{y}^n) \in \set{T}_\epsilon^n(X_{\set{A}} Y), \\
              &\phantom{-----------}
					\forall i \leq k_t, g_i(\hat{y}^n) = g_i(y^n)
				\}.
\end{align*}
Therefore, by the union bound, we have
\begin{align*}
	\expectation{
		\probability{
			\hat{Y}^n(\set{A}) \neq Y^n
		}
	} 
	&= 
	\expectation{
		\probability{
			\set{E}_0 \cup \set{E}_1
		}
	} \\
	&\leq 
	\expectation{
		\probability{
			\set{E}_0
		}
	} 
	+ 
	\expectation{
		\probability{
			\set{E}_1
		}
	},
\end{align*} 
where the expectation is over the choice of random binnings.\\
We now bound $\expectation{\probability{\set{E}_0}}$ as follows:
\begin{align}
	\expectation{\probability{\set{E}_0}} &= \expectation{\probability{(X_{\set{A}}^n,Y^n) \notin \typicalset{X_{\set{A}} Y}}} 
    \notag
    \\
	&= \expectation{ 1 - \probability{ (X_{\set{A}}^n,Y^n) \in \typicalset{X_{\set{A}} Y)} } }
    \notag
    \\
	&\leq \expectation{ 1 - (1 - \delta_{\epsilon}(n)) } 
    \notag
    \\
	\label{error 0 bound}
    &= \delta_{\epsilon}(n),
\end{align}
where the inequality follows from \cite[Corollary 2.1]{Physical_Layer_Security}.
Next, we bound $\expectation{\probability{\set{E}_1}}$ by
\begingroup
\allowdisplaybreaks
\begin{align}
	&\expectation{\probability{\set{E}_1}} 
	\notag
    \\
    &= \expectation{
			\probability{
				\exists \hat{y}^n \neq Y^n, 
				(X_{\set{A}}^n,\hat{y}^n) \in \typicalset{X_{\set{A}} Y}, 
                \notag
                \\
              &\phantom{--------------}
				\forall i \leq k_t, g_i(\hat{y}^n) = g_i(y^n)
			}
		}
    \notag
    \\
	&= \expectation[\Big]{
		\smash{\sum_{x_{\set{A}}^n, y^n} 
			p(x_{\set{A}}^n, y^n) 
			}\indicator{
				\exists \hat{y}^n \neq Y^n, 
				(X_{\set{A}}^n,\hat{y}^n) \in \typicalset{X_{\set{A}} Y}, 
                \notag
                \\
              &\phantom{--------------}
				\forall i \leq k_t, g_i(\hat{y}^n)=g_i(y^n)
            }
		}
    \notag
    \\
	&\leq \expectation[\Big]{
		\sum_{x_{\set{A}}^n, y^n} 
			p(x_{\set{A}}^n, y^n) 
            \mkern-30mu
			\sum_{
				\substack{
					\hat{y}^n \neq y^n \\ 
					\hat{y}^n \in \typicalset{X_{\set{A}} Y | x_{\set{A}}^n}
				}
			}
            \mkern-12mu
            \prod_{i=1}^{k_t}
			\indicator{
				g_i(\hat{y}^n)=g_i(y^n)
			}
		}
    \notag
	\\
	&= \sum_{x_{\set{A}}^n, y^n} 
		p(x_{\set{A}}^n, y^n)
        \sum_{
			\substack{
				\hat{y}^n \neq y^n \\ 
				\hat{y}^n \in \typicalset{X_{\set{A}} Y | x_{\set{A}}^n}
			}
		}
		\notag
        \\&\phantom{---------}\times
		\sum_{g^{k_t}}
		p(g^{k_t})
		\prod_{i=1}^{k_t}
		\indicator{
			g_i(\hat{y}^n)=g_i(y^n)
		}
    \notag
	\\
	&\note{=}
		\sum_{x_{\set{A}}^n, y^n} 
		p(x_{\set{A}}^n, y^n)
		\sum_{
			\substack{
				\hat{y}^n \neq y^n \\ 
				\hat{y}^n \in \typicalset{X_{\set{A}} Y | x_{\set{A}}^n}
			}
		}
		2^{-k_t n\epsilon}
    \notag
	\\
	&\leq 2^{-k_t n\epsilon}
		\sum_{x_{\set{A}}^n, y^n} 
		p(x_{\set{A}}^n, y^n)
		\cardinality[\big]{
			\typicalset{X_{\set{A}} Y | x_{\set{A}}^n}
		}
    \notag
	\\
	&\note{\leq} 2^{-k_t n\epsilon}
		\sum_{x_{\set{A}}^n, y^n} 
		p(x_{\set{A}}^n, y^n)
		2^{n(H(Y|X_{\set{A}}) + \delta(\epsilon))}
    \notag
	\\
	&= 2^{n(H(Y|X_{\set{A}}) + \delta(\epsilon) -k_t \epsilon )}
    \notag
	\\
    \label{error 1 bound}
	&\note{=} \delta_{\epsilon}(n),
\end{align}
\endgroup
where \refnote{2} follows from
\begin{align}
    \label{number of possible binnings}
	&\sum_{g^{k_t}}
		p(g^{k_t})
		\prod_{i=1}^{k_t}
		\indicator{
			g_i(\hat{y}^n) = g_i(y^n)
		}
    \\
	&= \prod_{i=1}^{k_t}
		\sum_{g_i}
		p(g_i)
		\indicator{
			g_i(\hat{y}^n) = g_i(y^n)
		}
    \notag
    \\
	&= \prod_{i=1}^{k_t}
		\sum_{g_i}
		\cardinality{ \set{M}_\epsilon } ^ { -\cardinality{\set{Y}^n} }
		\indicator{
			g_1(\hat{y}^n) = g_1(y^n)
		}
    \notag
	\\
	&= \prod_{i=1}^{k_t}
		\cardinality{ \set{M}_\epsilon } ^ { -\cardinality{\set{Y}^n} }
		\cardinality{ \set{M}_\epsilon } ^ { \cardinality{\set{Y}^n} - 1 }
    \notag
	\\
	&= \prod_{i=1}^{k_t}
		\cardinality{ \set{M}_\epsilon } ^ {-1} 
    \notag
	\\
	&= 2^{-k_t n\epsilon}, \notag
\end{align}
\refnote{1} follows from \cite[Theorem 2.2]{Physical_Layer_Security}, and \refnote{0} follows by choosing $k_t$ such that 
\begin{equation}
    \label{reliability constraint derived}
    k_t \epsilon > \max_{\set{A} \in \sset{A}_t} H(Y|X_{\set{A}}) + \delta(\epsilon).
\end{equation} 
From \eqref{error 0 bound} and \eqref{error 1 bound}, we have $\expectation{\probability{\hat{Y}^n(\set{A}) \neq Y^n}} \leq \delta_{\epsilon}(n)$.

\subsection{Proof of Secrecy} \label{Proof of Secrecy}
\setcounter{notecounter}{0}
For binnings $g^{b}$ and $h^{b}$ and $\set{U} \in \sset{U}_t$, such that $m^t = \concat_{i = 1}^{k_t} \mu_i$ and $s_t = \concat_{i = 1}^{\sigma_t} \nu_i$, we express the joint probability distribution $p_{M^t S_t X_{\set{U}}^n}$ as
\begin{align*}
	&p_{M^t S_t X_{\set{U}}^n} (m^t, s_t, x_{\set{U}}^n) \\
	&\note{=} \sum_{y^n} 
		p(x_{\set{U}}^n, y^n) 
		p(m^t|y^n) 
		p(s_t |y^n) \\
	&\note{=} \sum_{y^n} 
		p(x_{\set{U}}^n, y^n) 
		\prod_{i=1}^{k_t}
			\indicator{g_i(y^n) = \mu_i} 
		\prod_{j=1}^{\sigma_t}
			\indicator{h_j(y^n) = \nu_j},
\end{align*}
where \refnote{1} and \refnote{0} follow since the message and secret are uniquely determined by the dealer's observation and the choice of binnings and are therefore conditionally independent, from each other and from all other participant's observations, given the dealer's observation and the particular choice of binnings. 

Hence, on average over the choice of random binnings, we have:
\begin{align}
	&\expectation[\big]{p_{M^t S_t X_{\set{U}}^n} (m^t, s_t, x_{\set{U}}^n)} 
    \notag
    \\
	&= \sum_{g^{b},h^{b}} p(g^{b},h^{b}) 
		\sum_{y^n} 
			p(x_{\set{U}}^n, y^n) 
        \\&\phantom{-----}\times
			\prod_{i=1}^{k_t}
				\indicator{g_i(y^n) = \mu_i} 
			\prod_{j=1}^{\sigma_t}
				\indicator{h_j(y^n) = \nu_j} 
    \notag
	\\
	&= 2^{-n(\epsilon k_t + \epsilon \sigma_t)} 
		p_{X_{\set{U}}^n} (x_{\set{U}}^n)
    \notag
    \\
    \label{uniform message and secret}
	&= p_{M^t S_t}^{(\text{unif})} (m^t, s_t) 
        p_{X_{\set{U}}^n} (x_{\set{U}}^n),
\end{align}
where the second equality holds similar to \eqref{number of possible binnings}. Using \eqref{uniform message and secret}, we have
\begin{align}
	&\expectation[\big]{
		\variationaldistance
			{ p_{M^t S_t X_{\set{U}}^n} }
			{ p_{M^t S_t}^{(\text{unif})}p_{X_{\set{U}}^n}) }
	} \notag
    \\
	&= \expectation[\Big]{ 
		\smash{\sum_{m^t, s_t, x_{\set{U}}^n}} 
			\abs[\Big] {
				p_{M^t S_t X_{\set{U}}^n}(m^t, s_t, x_{\set{U}}^n) 
                \notag
        \\&\phantom{-----------}
				- p_{M^t S_t}^{(\text{unif})} (m^t, s_t) p_{X_{\set{U}}^n} (x_{\set{U}}^n)
			}
		} \notag
    \\
	&= \expectation[\Big]{ 
		\smash{\sum_{m^t, s_t, x_{\set{U}}^n}} 
			\abs[\Big]{ 
				p_{M^t S_t X_{\set{U}}^n}(m^t, s_t, x_{\set{U}}^n) 
                \notag
        \\&\phantom{-----------}
				 - \expectation[\big]{ p_{M^t S_t X_{\set{U}}^n}(m^t, s_t, x_{\set{U}}^n) }
			}
		} \notag
    \\
    \label{expectation of variational distance}
	\begin{split}
	&=\smash{\sum_{k=1}^2} 
		\expectation[\Big]{
			\smash{\sum_{m^t, s_t, x_{\set{U}}^n} }
				\abs[\big]{ 
					p_{M^t S_t X_{\set{U}}^n}^{(k)} (m^t, s_t, x_{\set{U}}^n)  
        \\&\phantom{----------}
			 		- \expectation[\big]{ p_{M^t S_t X_{\set{U}}^n}^{(k)} (m^t, s_t, x_{\set{U}}^n) }
				}
		},
    \end{split}
\end{align}
where we break up \eqref{expectation of variational distance} into two cases, one where the dealer's observation $y^n$ is jointly typical with the unauthorized group's observation $x_{\set{U}}^n$ and one where these observations are not jointly typical: 
\begin{align}
    \label{typical probability distribution}
	&p_{M^t S_t X_{\set{U}}^n}^{(1)}(m^t, s_t, x_{\set{U}}^n) = 
		\sum_{y^n \in \typicalset{X_{\set{U}} Y | x_{\set{U}}^n} } 
			p(x_{\set{U}}^n, y^n) 
            \notag
            \\&\phantom{--}\times
			\prod_{i=1}^{k_t}
				\indicator{g_i(y^n) = \mu_i} 
			\prod_{j=1}^{\sigma_t}
				\indicator{h_j(y^n) = \nu_j} 
	\\
    \label{atypical probability distribution}
	&p_{M^t S_t X_{\set{U}}^n}^{(2)}(m^t, s_t, x_{\set{U}}^n) = 
		\sum_{y^n \notin \typicalset{X_{\set{U}} Y | x_{\set{U}}^n} }
			p(x_{\set{U}}^n, y^n) 
            \notag
            \\&\phantom{--}\times
		\prod_{i=1}^{k_t}
				\indicator{g_i(y^n) = \mu_i} 
			\prod_{j=1}^{\sigma_t}
				\indicator{h_j(y^n) = \nu_j}.
\end{align}
For the case where the observations are typical, we bound the corresponding term in \eqref{expectation of variational distance} as follows:

\begin{align}
	&\expectation[\Big]{
		\sum_{m^t, s_t, x_{\set{U}}^n} 
			\abs[\big]{ 
				 p_{M^t S_t X_{\set{U}}^n}^{(1)} (m^t, s_t, x_{\set{U}}^n) 
				 - \expectation[\big]{ p_{M^t S_t X_{\set{U}}^n}^{(1)} (m^t, s_t, x_{\set{U}}^n) }
			}
	} \notag
	\\
    \label{sqrt variance equation}
	&\leq \sum_{m^t, s_t, x_{\set{U}}^n}
		\sqrt{
			\variance{ 
				 p_{M^t S_t X_{\set{U}}^n}^{(1)} (m^t, s_t, x_{\set{U}}^n) 
			}
		},
\end{align}
where the inequality holds by Jensen's Inequality.
Isolating the variance in \eqref{sqrt variance equation} and substituting in \eqref{typical probability distribution} we have
\begin{align}
	&\variance[\Big]{
		\sum_{y^n \in \typicalset{X_{\set{U}} Y | x_{\set{U}}^n}}
			p(x_{\set{U}}^n, y^n) 
            \notag
            \\&\phantom{--}\times
			\prod_{i=1}^{k_t}
				\indicator{g_i(y^n) = \mu_i} 
			\prod_{j=1}^{\sigma_t}
				\indicator{h_j(y^n) = \nu_j} 
	}
    \notag
	\\
	&\note{=} \sum_{y^n \in \typicalset{X_{\set{U}} Y | x_{\set{U}}^n}}
			p^2(x_{\set{U}}^n, y^n)  
            \notag
            \\&\phantom{--}\times
			\variance[\Big]{
				\prod_{i=1}^{k_t}
					\indicator{g_i(y^n) = \mu_i} 
				\prod_{j=1}^{\sigma_t}
					\indicator{h_j(y^n) = \nu_j}
			}
    \notag
	\\
	&\leq \sum_{y^n \in \typicalset{X_{\set{U}} Y | x_{\set{U}}^n}}
			p^2(x_{\set{U}}^n, y^n)  
            \notag
            \\&\phantom{--}\times
			\expectation[\Big]{
				\sqrd[\big]{
					\prod_{i=1}^{k_t}
						\indicator{g_i(y^n) = \mu_i} 
					\prod_{j=1}^{\sigma_t}
						\indicator{h_j(y^n) = \nu_j}
				}
			}
    \notag
	\\
	&= \sum_{y^n \in \typicalset{X_{\set{U}} Y | x_{\set{U}}^n}} 
			p^2(x_{\set{U}}^n, y^n)  
            \notag
            \\&\phantom{--}\times
			\expectation[\Big]{
				\prod_{i=1}^{k_t}
					\indicator{g_i(y^n) = m_i} 
				\prod_{j=1}^{\sigma_t}
					\indicator{h_j(y^n) = s_j}
			}
    \notag
	\\
	&\note{=} \sum_{y^n \in \typicalset{X_{\set{U}} Y | x_{\set{U}}^n}} 
			p^2(x_{\set{U}}^n, y^n) 
			2^{-n(\epsilon \sigma_t + \epsilon k_t)}
    \notag
	\\
	&=  2^{-n(\epsilon \sigma_t + \epsilon k_t)} 
		p^2(x_{\set{U}}^n)
		\sum_{y^n \in \typicalset{X_{\set{U}} Y | x_{\set{U}}^n}} 
			p^2(y^n|x_{\set{U}}^n)
    \notag
	\\
	&\note{\leq} 2^{-n(\epsilon \sigma_t + \epsilon k_t)}
		  p^2(x_{\set{U}}^n)
		\sum_{y^n \in \typicalset{X_{\set{U}} Y | x_{\set{U}}^n}} 
			2^{-2n(H(Y|X_{\set{U}}) - \delta(\epsilon))}
    \notag
	\\
	&=  2^{-n(\epsilon \sigma_t + \epsilon k_t)}
		p^2(x_{\set{U}}^n) 
		\cardinality[\big]{		
			\typicalset{X_{\set{U}} Y | x_{\set{U}}^n}
		}
		2^{-2n(H(Y|X_{\set{U}}) - \delta(\epsilon))}
    \notag
	\\
	&\note{\leq} 2^{-n(\epsilon \sigma_t + \epsilon k_t)}
		 p^2(x_{\set{U}}^n) 
		 2^{n(H(Y|X_{\set{U}}) + \delta(\epsilon))} 
		 2^{-2n(H(Y|X_{\set{U}}) - \delta(\epsilon))}
    \notag
	\\
    \label{Isolated Variance}
	&= p^2(x_{\set{U}}^n) 2^{-n(H(Y|X_{\set{U}}) - 3\delta(\epsilon) + (\epsilon \sigma_t + \epsilon k_t) },
\end{align}
where \refnote{3} holds because $p(x_{\set{U}}^n, y^n)$ is a constant and all of the random binnings are mutually independent, \refnote{2} holds similar \mbox{to \eqref{number of possible binnings}}, and \refnote{1} and \refnote{0} follow from \cite[Theorem 2.2]{Physical_Layer_Security}. \mbox{Substituting \eqref{Isolated Variance}} back into \eqref{sqrt variance equation}, we have
\begin{align*}
	&\sum_{m^t, s_t, x_{\set{U}}^n} 
		\sqrt{
			\variance{
				p_{M^t S_t X_{\set{U}}^n}^{(1)}(m^t, s_t, x_{\set{U}}^n)
			}
		} 
	\\
	&\leq \sum_{m^t, s_t, x_{\set{U}}^n} 
		\sqrt{
			p^2(x_{\set{U}}^n) 
			2^{
				- n(H(Y|X_{\set{U}}) 
				- 3\delta(\epsilon) 
				+ \epsilon \sigma_t
				+ \epsilon k_t 
			}
		}
	\\
	&= 2^{
			-\frac{n}{2} (
				H(Y|X_{\set{U}}) 
				- 3\delta(\epsilon) 
				+ \epsilon \sigma_t 
				+ \epsilon k_t
			) 
		} 
		\sum_{m^t, s_t, x_{\set{U}}^n} 
			p(x_{\set{U}}^n) 
	\\
	&= 2^{
			-\frac{n}{2} (
				H(Y|X_{\set{U}}) 
				- 3\delta(\epsilon) 
				+ \epsilon \sigma_t
				+ \epsilon k_t
			) 
		} 
		2^{
			n(\epsilon \sigma_t + \epsilon k_t)	
		}
	\\
	&= 2^{
			-\frac{n}{2} (
				H(Y|X_{\set{U}}) 
				- 3\delta(\epsilon) 
				- (\epsilon \sigma_t + \epsilon k_t )
			) 
		} 
	\\
	&\leq \delta_\epsilon(n),
\end{align*}
where the last inequality follows from choosing $\sigma_t$ and $k_t$ that satisfy
\begin{equation}
    \label{secrecy constraint derived}
    \epsilon \sigma_t + \epsilon k_t < 
    \min_{ \set{U} \in \sset{U}_t } H(Y|X_{\set{U}}) - 3\delta(\epsilon).
\end{equation}
For the case where the observations are not jointly typical, we bound the corresponding term in \eqref{expectation of variational distance} as follows:\\
\begin{align*}
	&\expectation[\Big]{
		\sum_{m^t, s_t, x_{\set{U}}^n} \!
			\abs[\big]{
				p_{M^t S_t X_{\set{U}}^n}^{(2)} (m^t, s_t, x_{\set{U}}^n) 
				- 
				\expectation[\big]{ 
					p_{M^t S_t X_{\set{U}}^n}^{(2)} (m^t, s_t, x_{\set{U}}^n) 
				}
			}
	} 
	\\
	&\note{\leq} \sum_{m^t, s_t, x_{\set{U}}^n}
		2
		\expectation[\big]{ 
			p_{M^t S_t X_{\set{U}}^n}^{(2)} (m^t, s_t, x_{\set{U}}^n)
		}
	\\
	&\note{=} 2 \sum_{m^t, s_t, x_{\set{U}}^n} 
		\sum_{y^n \notin \typicalset{X_{\set{U}} Y | x_{\set{U}}^n}}
			p(x_{\set{U}}^n, y^n)
			2^{-n(\epsilon \sigma_t + \epsilon k_t )}
	\\
	&= 2 \sum_{m^t, s_t}
			2^{-n(\epsilon \sigma_t + \epsilon k_t )}
		 	\probability{
		 		(X_{\set{U}}^n Y^n) \notin \typicalset{X_{\set{U}} Y}
		 	}
	\\
	&= 2 \probability{
			(X_{\set{U}}^n Y^n) \notin \typicalset{X_{\set{U}} Y}
		 }
	\\
	&\note{\leq} 2\delta_{\epsilon}(n),
\end{align*}
where \refnote{1} is by application of the triangle inequality, \refnote{0} holds similar to \eqref{number of possible binnings}, and \refnote{0} follows from \cite[Corollary 2.1]{Physical_Layer_Security}. Therefore,
$\expectation{
	\variationaldistance
		{ p_{M S X_{\set{U}}^n} } 
		{ p_{MS}^{(\text{unif})}p_{X_{\set{U}}^n}) } 
}
\leq \delta_{\epsilon}(n)$.
\setcounter{notecounter}{0}

This means for a fixed access structure $(\sset{A}_t,\sset{U}_t)$, by choosing $k_t(\sset{A}_t,\sset{U}_t)$ and $\sigma_t(\sset{A}_t,\sset{U}_t)$ that satisfy both \eqref{reliability constraint derived} and \eqref{secrecy constraint derived} it will be the case that
\begin{align}
\label{expected reliability bound}
    \expectation{
        \probability{
            \hat{Y}^n(\set{\set{A}}) \neq Y^n
        }
    }
    \leq
    \delta_{\epsilon}(n),\\
\label{expected variational distance bound}
    \expectation{
        \variationaldistance
            {p_{M^t S_t X_{\set{U}}^n}} 
            {p_{M^t S_t}^{(\text{unif})}p_{X_{\set{U}}^n}} 
    }
    \leq 
    \delta_{\epsilon}(n).
\end{align}
Using \eqref{expected reliability bound} and \eqref{expected variational distance bound}, we have
\begin{align}
    &\max_{\set{A} \in \sset{A}_t} 
        \expectation[\Big]{
            \probability{\hat{Y}^n(\set{A}) \neq Y^n } 
        }
    \notag
    \\&\phantom{-----}
    +
    \max_{\set{U} \in \sset{U}_t} 
        \expectation[\Big]{
            \variationaldistance[\big]
                {p_{M^t S_t X_{\set{U}}^n}} 
                {p_{M^t S_t}^{(\text{unif})}p_{X_{\set{U}}^n}}
        }
    \notag
    \\
    &\leq
    \smash{\sum_{\set{A} \in \sset{A}_t}} 
        \expectation[\Big]{
            \probability{\hat{Y}^n(\set{A}) \neq Y^n } 
        }
    \notag
    \\&\phantom{-----}
    +
    \sum_{\set{U} \in \sset{U}_t} 
        \expectation[\Big]{
            \variationaldistance[\big]
                {p_{M^t S_t X_{\set{U}}^n}} 
                {p_{M^t S_t}^{(\text{unif})}p_{X_{\set{U}}^n}}
        }
    \notag
    \\
    \label{bound for fixed access structure}
    &\leq
    \sum_{\set{A} \in \sset{A}_t} 
        \delta_{\epsilon}(n)
    +
    \sum_{\set{U} \in \sset{U}_t} 
        \delta_{\epsilon}(n).
\end{align}
To extend these bounds to any AAS, we show that the worst case over every time step and access structure is still bounded by $\delta_{\epsilon}(n)$. We have
\begin{align*}
	&\expectation[\bigg]{
		\sum_{t \in \set{T}} 
			\sum_{\sset{A}_t \subseteq 2^{\set{L}}}
				\sum_{\sset{U}_t \subseteq 2^{\set{L}} - \sset{A}_t}
				\parenthesis[\Big]{
					\max_{\set{A} \in \sset{A}_t} 
						\probability{\hat{Y}^n(A) \neq Y^n } 
                    \\&\phantom{--}
					+
					\max_{\set{U} \in \sset{U}_t} 
						\variationaldistance[\big]
							{p_{M^t S_t X_{\set{U}}^n}} 
							{p_{M^t S_t}^{(\text{unif})}p_{X_{\set{U}}^n}}		
                }
	}
	\\
	&\leq
	\expectation[\bigg]{
		\sum_{t \in \set{T}} 
			\sum_{\sset{A}_t \subseteq 2^{\set{L}}}
				\sum_{\sset{U}_t \subseteq 2^{\set{L}} - \sset{A}_t}
				\parenthesis[\Big]{
					\sum_{\set{A} \in \sset{A}_t} 
						\probability{\hat{Y}^n(A) \neq Y^n } 
                    \\&\phantom{--}
					+
					\sum_{\set{U} \in \sset{U}_t} 
						\variationaldistance[\big]
							{p_{M^t S_t X_{\set{U}}^n}} 
							{p_{M^t S_t}^{(\text{unif})}p_{X_{\set{U}}^n}}
				}
	}
	\\
	&=
	\sum_{t \in \set{T}} 
		\sum_{\sset{A}_t \subseteq 2^{\set{L}}}
			\sum_{\sset{U}_t \subseteq 2^{\set{L}} - \sset{A}_t}
			\parenthesis[\Big]{
				\sum_{\set{A} \in \sset{A}_t} 
					\expectation[\Big]{
						\probability{\hat{Y}^n(A) \neq Y^n } 
					}
				\\&\phantom{--}
                +
				\sum_{\set{U} \in \sset{U}_t} 
					\expectation[\Big]{
						\variationaldistance[\big]
							{p_{M^t S_t X_{\set{U}}^n}} 
							{p_{M^t S_t}^{(\text{unif})}p_{X_{\set{U}}^n}}
					}
			}
	\\
	&\leq
	\sum_{t \in \set{T}} 
		\sum_{\sset{A}_t \subseteq 2^{\set{L}}}
			\sum_{\sset{U}_t \subseteq 2^{\set{L}} - \sset{A}_t}
			\left(
				\sum_{\set{A} \in \sset{A}_t} 
					\delta_{\epsilon}(n)
				+
				\sum_{\set{U} \in \sset{U}_t} 
					\delta_{\epsilon}(n)
			\right)
	\\
	&\leq \delta_{\epsilon}(n),
\end{align*}
where the second  inequality holds from the choices of $k_t(\sset{A}_t,\sset{U}_t)$ and $\sigma_t(\sset{A}_t,\sset{U}_t)$ used in \eqref{bound for fixed access structure} for each $\sset{A}_t$ and $\sset{U}_t$. Since the expectation of the constraints over all the binnings is upper bounded by $\delta_{\epsilon}(n)$, using \cite[Lemma 2.2]{Physical_Layer_Security} shows there exists a particular choice of binnings $g^{b}$ and $h^{b}$ that for every possible AAS, at each time step $t \in \set{T}$,
\begin{align}
\label{Reliability Constraint}
    \max_{\set{A} \in \sset{A}_t} 
	\sset{P}[\hat{Y}^n(\set{A}) \neq Y^n] 
	&\leq \delta_{\epsilon}(n), \\
\label{Secrecy Constraint}
    \max_{\set{U} \in \sset{U}_t} 
	\variationaldistance{p_{M^t S_t X_{\set{U}}^n}}{ p_{M^t S_t}^{(\text{unif})}p_{X_{\set{U}}^n}} 
	&\leq \delta_{\epsilon}(n) 
\end{align}
are satisfied. The reliability condition of Definition \ref{Achievability Definition} is satisfied for every $t \in \set{T}$ since
\begin{align*}
\max_{\set{A} \in \sset{A}_t} 
		\sset{P}[\hat{S}_t \neq S_t]  \leq \max_{\set{A} \in \sset{A}_t} \sset{P}[\hat{Y}^n(\set{A}) \neq Y^n] 
       .
\end{align*}
For the secrecy condition of Definition \ref{Achievability Definition}, first note that \mbox{using \cite[Lemma 16]{Cuff_Thesis}},
$\variationaldistance{p_{M^t S_t X_{\set{U}}^n}}{ p_{M^t S_t}^{(\text{unif})}p_{X_{\set{U}}^n}} 
		\leq \delta_{\epsilon}(n)$
implies both
$\variationaldistance{p_{M^t X_{\set{U}}^n}}{ p_{M^t}^{(\text{unif})}p_{X_{\set{U}}^n}}
    \leq \delta_{\epsilon}(n)
$ and 
$\variationaldistance{p_{S_t}}{ p_{S_t}^{(\text{unif})}}
    \leq \delta_{\epsilon}(n)
$, from which we have for every $t \in \set{T}$
\begin{align}	
    \label{variational distance bound}
	&\max_{\set{U} \in \sset{U}_t} 
		\variationaldistance{p_{M^t S_t X_{\set{U}}^n}}{ p_{S_t}p_{M^t X_{\set{U}}^n}}
    \\
	&\note{\leq} \max_{\set{U} \in \sset{U}_t} 
		[\variationaldistance{p_{M^t S_t X_{\set{U}}^n}}{ p_{M^t S_t}^{(\text{unif})}p_{X_{\set{U}}^n}} \! + \!
		\variationaldistance{p_{M^t S_t}^{(\text{unif})}p_{X_{\set{U}}^n}}{ p_{S_t} p_{M^t X_{\set{U}}^n}}] 
    \notag
    \\
	&\note{\leq} \delta_{\epsilon}(n) + 
		\max_{\set{U} \in \sset{U}_t} 
			[\variationaldistance{p_{M^t S_t}^{(\text{unif})}p_{X_{\set{U}}^n}}{ p_{S_t}^{\text{(unif)}} p_{M^t X_{\set{U}}^n}} 
    \notag
    \\
    &\phantom{~\leq \delta_{\epsilon}(n) }    
            + 
			\variationaldistance{p_{S_t}^{(\text{unif})}p_{M^t X_{\set{U}}^n}}{ p_{S_t} p_{M^t X_{\set{U}}^n}}]
    \notag
    \\
	&= \delta_{\epsilon}(n) + 
		\max_{\set{U} \in \sset{U}_t} 
			\variationaldistance{p_{M^t}^{(\text{unif})}p_{X_{\set{U}}^n}}{ p_{M^t X_{\set{U}}^n}} +
			\variationaldistance{p_{S_t}^{(\text{unif})}}{ p_{S_t}}
    \notag
    \\
	&\leq \delta_{\epsilon}(n) + 
		\delta_{\epsilon}(n) + 
		\delta_{\epsilon}(n) 
    \notag
    \\
	&= \delta_{\epsilon}(n),
    \notag
\end{align}
where \refnote{1} and \refnote{0} follow from the triangle inequality.
Then, from \eqref{variational distance bound}, using \cite[Inequality 2]{On_Estimation_of_Information_via_Variation}, we get
$\max_{\set{U} \in \sset{U}_t} I (S_t;M^t,X_{\set{U}}^n) \leq \delta_{\epsilon}(n)$ for every $t \in \set{T}$. 
Finally, the secret uniformity condition of Definition \ref{Achievability Definition} holds because for every $t \in \set{T}$ we have $\variationaldistance{p_{S_t}}{p_{S_t}^{(\text{unif})}} \leq \delta_{\epsilon}(n)$ which implies $H(p_{S_t}^{(\text{unif})}) - H(p_{S_t}) \leq \delta_{\epsilon}(n)$ using \cite[Lemma 2.7]{Csiszár_Textbook}{ and therefore $\log|S_t| - H(S_t) \leq \delta_{\epsilon}(n)$.

\subsection{Achievable Rates} \label{Achievable Rates}
\setcounter{notecounter}{0}
Sections \ref{Proof of Reliability} and \ref{Proof of Secrecy} prove that a rate choice that satisfies 
\begin{align*}
	\epsilon k_t &> \max_{\set{A} \in \sset{A}_t} H(Y|X_{\set{A}}), \\
	\epsilon \sigma_t + \epsilon k_t &< \min_{\set{U} \in \sset{U}_t} H(Y|X_{\set{U}})
\end{align*}
is achievable.
At time $t$, the dealer  chooses 
\begin{align}
\label{k_t choice}
    k_t &= \ceil[\Bigg]{
        \frac{\displaystyle\max_{\set{A} \in \sset{A}_t} H(Y|X_{\set{A}}) + \delta(\epsilon)}{\epsilon}
    }, \\
\label{sigma_t choice}
    \sigma_t &= \floor[\Bigg]{
        \frac{\displaystyle\min_{\set{U} \in \sset{U}_t} H(Y|X_{\set{U}}) - \epsilon k_t - \delta(\epsilon)}{\epsilon}
    }.
\end{align}

Then, define the message rate $R^\prime_t \triangleq (k_t - k_{t-1})\epsilon$ and secret rate $R_t \triangleq \sigma_t \epsilon$ at time $t$. We have
\begin{align}
    \lim_{\epsilon \rightarrow 0} \sum_{i=1}^{t} R^\prime_i 
    &= \lim_{\epsilon \rightarrow 0} \sum_{i=1}^{t} (k_i - k_{i-1})\epsilon 
    \notag
    \\
    &= \lim_{\epsilon \rightarrow 0} k_t\epsilon  
    \notag
    \\
    \label{Message rate}
    &\leq \max_{\set{A} \in \sset{A}_t} H(Y|X_{\set{A}}),
\end{align}
where the last inequality holds by \eqref{k_t choice}. Then, using \eqref{sigma_t choice} \mbox{and \eqref{Message rate}}, we have
\begin{align*}
    \lim_{\epsilon \rightarrow 0} R_t 
    &= \lim_{\epsilon \rightarrow 0} \sigma_t \epsilon
    \\
    &\geq \min_{\set{U} \in \sset{U}_t} H(Y|X_{\set{U}}) 
        - \max_{\set{A} \in \sset{A}_t} H(Y|X_{\set{A}}) .
\end{align*}

\section{Concluding Remarks} \label{Conclusion} 
We generalize secret-sharing models that rely on correlated randomness and public communication, originally designed for a fixed access structure, to work for an Additive Access Structure. For this model, we have proven the existence of a secret sharing strategy that achieves the same secret rate at each time step as the the best known strategy for a fixed access structure. We have also proven that there exists a strategy that achieves the secret sharing capacity at any time step where the access structure is a threshold access structure. In that case, this proves that, counterintuitively, knowing the AAS in advance does not provide a benefit to the achievable secret rates.  While constructive achievability proofs have been proposed for secret sharing with fixed access structure, e.g., \mbox{\cite{Low-Complexity_Secret_Sharing_Schemes_Using_Correlated_Random_Variables_and_Rate-Limited_Public_Communication,
Secret_Sharing_Over_a_Gaussian_Broadcast_Channel_Optimal_Coding_Scheme_Design_and_Deep_Learning_Approach_at_Short_Blocklength,
sultana2025secretsharingschemescorrelated,
Explicit_Wiretap_Channel_Codes_via_Source_Coding_Universal_Hashing_and_Distribution_Approximation_When_the_Channels_Statistics_are_Uncertain}}, providing a constructive achievability proof for secret sharing with AAS remains an open problem. 

\appendix
We first prove the following result where Theorems \ref{Achievability Result} and~\ref{Converse Result} match at a given time $t$. 

\begin{proposition}
\label{proposition:1}
 Given an AAS, at any time $t$ where  $(\sset{A}_t, \sset{U}_t)$ form a threshold access structure $\sset{A}_t = \{\set{A} \subseteq \set{L} : \cardinality{\set{A}} \geq u\}$ and $\sset{U}_t = \{\set{U} \subseteq \set{L} : \cardinality{\set{U}} \leq v\}$, the secret capacity is $R_t = H(X) (u-v)$.
\end{proposition}
\begin{proof}
First, we show that $R_t = H(X) (u-v)$ is an achievable rate. 
By Theorem \ref{Achievability Result}, 
    \begin{align*}
        R_t &= \min_{\set{U} \in \sset{U}_t} H(Y|X_{\set{U}}) 
        - \max_{\set{A} \in \sset{A}_t} H(Y|X_{\set{A}})
        \\
        &= \min_{\set{U} \in \sset{U}_t} H(X_{\set{L}}|X_{\set{U}}) 
        - \max_{\set{A} \in \sset{A}_t} H(X_{\set{L}}|X_{\set{A}})
        \\
        &= \min_{\set{U} \in \sset{U}_t} H(X_{\set{L} \backslash \set{U}}) 
        - \max_{\set{A} \in \sset{A}_t} H(X_{\set{L} \backslash \set{A}})
        \\
        &= H(X)(L-v)
        - H(X)(L-u)
        \\
        &= H(X)(u-v),
    \end{align*}
where the fourth equality follows from the independence of the private keys.
Next, we show that this rate is optimal.
By Theorem \ref{Converse Result}, 
    \begin{align*}
        R_t &\leq \min_{\set{U} \in \sset{U}_t} \min_{\set{A} \in \sset{A}_t} 
                I(Y;X_{\set{A}}|X_{\set{U}}) \\
            &= \min_{\set{U} \in \sset{U}_t} \min_{\set{A} \in \sset{A}_t} 
                I(X_{\set{L}};X_{\set{A}}|X_{\set{U}}) \\
            &= \min_{\set{U} \in \sset{U}_t} \min_{\set{A} \in \sset{A}_t} 
                H(X_{\set{L}}|X_{\set{U}}) -
                H(X_{\set{L}}|X_{\set{A}},X_{\set{U}}) \\
            &= \min_{\set{U} \in \sset{U}_t} \min_{\set{A} \in \sset{A}_t} 
                H(X_{\set{L} \backslash \set{U}}) -
                H(X_{\set{L} \backslash (\set{A} \cup \set{U})}) \\
            &= \min_{\set{U} \in \sset{U}_t} \min_{\set{A} \in \sset{A}_t} 
                H(X)(\cardinality{\set{L} \backslash \set{U}}) -
                H(X)(\cardinality{\set{L} \backslash (\set{A} \cup \set{U})}) \\
            &= \min_{\set{U} \in \sset{U}_t} \min_{\set{A} \in \sset{A}_t} 
                H(X)(\cardinality{\set{L}}-\cardinality{\set{U}}) -
                H(X)(\cardinality{\set{L}}-\cardinality{\set{A} \cup \set{U}}) \\
            &= \min_{\set{U} \in \sset{U}_t} \min_{\set{A} \in \sset{A}_t} 
                H(X)(\cardinality{\set{A} \cup \set{U}} -\cardinality{\set{U}}) \\
            &= \min_{\set{U} \in \sset{U}_t} \min_{\set{A} \in \sset{A}_t} 
                H(X)(\cardinality{\set{A}} -\cardinality{\set{A} \cap \set{U}}),
    \end{align*}
    which is minimized when $\cardinality{\set{A}} = u$ and $\set{U} \subseteq \set{A}$, giving $R_t \leq H(X)(u-v)$.

\end{proof}
Finally, from Proposition \ref{proposition:1}, we deduce Theorem \ref{proposition:2}.

\bibliographystyle{ieeetr}
\bibliography{abbreviations,references}

@STRING{ACM_M_C           = "Commun. {ACM}"}

@STRING{ACM_C_STC         = "Proc. 20th Annu. {ACM} Symp. Theory Comput."}

@STRING{SPRINGER_J_PIT    = "Probl. Inf. Trans."}

@STRING{SPRINGER_C_ICDCN  = "Proc. Int. Conf. Distrib. Comput. Netw."}

@STRING{SPRINGER_C_ICCC   = "Proc. Int. Conf. Coding Cryptol"}

@STRING{SPRINGER_C_TCC    = "Proc. Theory Cryptogr. Conf."}

@STRING{SPRINGER_C_ISA    = "Inf. Secur. Appl."}

@STRING{SPRINGER_C_ICTACIS = "Proc. Int. Conf. Theory Appl. Cryptol. Inf. Secur."}

@STRING{IEEE_C_ISIT       = "Proc. {IEEE} Int. Symp. Inf. Theory"}

@STRING{IEEE_C_SFCS       = "Proc. 24th Annu. Symp. Found. Comput. Sci. ({SFCS})"}

@STRING{IEEE_C_SPAWC      = "Proc. {IEEE} 14th Workshop Signal Process. Adv. Wireless Commun. ({SPAWC})"}

@STRING{IEEE_J_IFS        = "{IEEE} Trans. Inf. Forensics Security"}

@STRING{IEEE_J_IT         = "{IEEE} Trans. Inf. Theory"}

@book{Csiszár_Textbook,
    author={Csisz{\'a}r, I. and K{\"o}rner, J.},
    title={Information Theory: Coding Theorems for Discrete Memoryless Systems},
    publisher={Cambridge University Press},
    year={2011},
}

@book{Physical_Layer_Security,
    author = {Bloch, Matthieu and Barros, Joo},
    title = {Physical-Layer Security: From Information Theory to Security Engineering},
    publisher = {Cambridge University Press},
    year = {2011},
    isbn = {0521817951},
    address = {USA},
    edition = {1st},
}

@article{rana2021information,
  title={Information-theoretic secret sharing from correlated {G}aussian random variables and public communication},
  author={Rana, Vidhi and Chou, R{\'e}mi A and Kwon, Hyuck M},
    journal=IEEE_J_IT,
  volume={68},
  number={1},
  pages={549--559},
  year={2021}}

@article{Distributed_Secret_Sharing,
    author={Chou, R\'{e}mi A.},
    title={Distributed Secret Sharing Over a Public Channel From Correlated Random Variables}, 
    journal=IEEE_J_IT,
    year={2024},
    volume={70},
    number={4},
    pages={2851-2869},
    doi={10.1109/TIT.2024.3363644}
}

@article{Explicit_Wiretap_Channel_Codes_via_Source_Coding_Universal_Hashing_and_Distribution_Approximation_When_the_Channels_Statistics_are_Uncertain,
    author={Chou, R\'{e}mi A.},
    title={Explicit Wiretap Channel Codes via Source Coding, Universal Hashing, and Distribution Approximation, When the Channels’ Statistics are Uncertain}, 
    journal=IEEE_J_IFS, 
    year={2023},
    volume={18},
    pages={117-132},
    doi={10.1109/TIFS.2022.3218414}
}

@article{How_to_share_a_secret_infinitely,
    author={Komargodski, Ilan and Naor, Moni and Yogev, Eylon},
    title={How to Share a Secret, Infinitely},
    journal=IEEE_J_IT, 
    year={2018},
    volume={64},
    number={6},
    pages={4179-4190},
    doi={10.1109/TIT.2017.2779121}
}

@article{On_Estimation_of_Information_via_Variation,
    author = {Pinsker, M.},
    title = {On Estimation of Information via Variation},
    journal = SPRINGER_J_PIT, 
    year = {2005},
    month = {04},
    publisher={Springer},
    pages = {71-75},
    volume = {41},
    doi = {10.1007/s11122-005-0012-8}
}

@article{Channel_Identification_Secret_Sharing_Using_Reciprocity_in_Ultrawideband_Channels,
    author={Wilson, Robert and Tse, David and Scholtz, Robert A.},
    title={Channel Identification: Secret Sharing Using Reciprocity in Ultrawideband Channels}, 
    journal=IEEE_J_IFS,
    year={2007},
    volume={2},
    number={3},
    pages={364-375},
    doi={10.1109/TIFS.2007.902666}
}

@article{Automatic_Secret_Keys_From_Reciprocal_MIMO_Wireless_Channels_Measurement_and_Analysis,
    author={Wallace, Jon W. and Sharma, Rajesh K.},
    title={Automatic Secret Keys From Reciprocal {MIMO} Wireless Channels: Measurement and Analysis}, 
    journal=IEEE_J_IFS, 
    year={2010},
    volume={5},
    number={3},
    pages={381-392},
    doi={10.1109/TIFS.2010.2052253}
}

@article{Information-Theoretically_Secret_Key_Generation_for_Fading_Wireless_Channels,
    author={Ye, Chunxuan and Mathur, Suhas and Reznik, Alex and Shah, Yogendra and Trappe, Wade and Mandayam, Narayan B.},
    title={Information-Theoretically Secret Key Generation for Fading Wireless Channels}, 
    journal=IEEE_J_IFS, 
    year={2010},
    volume={5},
    number={2},
    pages={240-254},
    doi={10.1109/TIFS.2010.2043187}
}

@article{Shamir_Secret_Sharing,
    author = {Shamir, Adi},
    title = {How to share a secret},
    journal = ACM_M_C,
    year = {1979},
    issue_date = {Nov. 1979},
    publisher = {Association for Computing Machinery},
    address = {New York, NY, USA},
    volume = {22},
    number = {11},
    issn = {0001-0782},
    url = {https://doi.org/10.1145/359168.359176},
    doi = {10.1145/359168.359176},
    pages = {612–613},
    numpages = {2},
}

@article{zou2015information,
  title={An information theoretic approach to secret sharing},
  author={Zou, Shaofeng and Liang, Yingbin and Lai, Lifeng and Shamai, Shlomo},
  journal=IEEE_J_IT,
  volume={61},
  number={6},
  pages={3121-3136},
  year={2015},
  publisher={IEEE}
}

@article{sultana2025secretsharingschemescorrelated,
    author={Sultana, Rumia and Chou, R\'{e}mi A.},
    title={Secret Sharing Schemes from Correlated Random Variables and Rate-Limited Public Communication},
    journal={arXiv preprint arXiv:2505.04076},
    year={2025}
}

@InProceedings{Secret_Sharing_Survey,
    author={Beimel, Amos},
    title={Secret-sharing schemes: A survey},
    booktitle=SPRINGER_C_ICCC,
    pages={11-46},
    year={2011},
    organization={Springer}
}

@InProceedings{Multiparty_unconditionally_secure_protocols,
    author = {Chaum, David and Cr\'{e}peau, Claude and Damgard, Ivan},
    title = {Multiparty unconditionally secure protocols},
    booktitle = ACM_C_STC,
    year = {1988},
    isbn = {0897912640},
    url = {https://doi.org/10.1145/62212.62214},
    doi = {10.1145/62212.62214},
    pages = {11–19},
    numpages = {9},
}

@InProceedings{byzantine_generals,
    author={Rabin, Michael O.},
    title={Randomized byzantine generals}, 
    booktitle=IEEE_C_SFCS, 
    year={1983},
    pages={403-409},
    doi={10.1109/SFCS.1983.48}
}

@InProceedings{Shared_generation_of_authenticators_and_signatures,
    author={Desmedt, Yvo and Frankel, Yair},
    title={Shared generation of authenticators and signatures},
    booktitle={CRYPTO},
    year={1991},
    pages={457-469},
}

@InProceedings{Alternative_Protocols_for_Generalized_Oblivious_Transfer,
    title={Alternative protocols for generalized oblivious transfer},
    author={Shankar, Bhavani and Srinathan, Kannan and Rangan, C Pandu},
    booktitle=SPRINGER_C_ICDCN,
    pages={304-309},
    year={2008},
    organization={Springer}
}

@InProceedings{Experimental_aspects_of_secret_key_generation_in_indoor_wireless_environments,
    author={Pierrot, Alexandre J. and Chou, R\'{e}mi A. and Bloch, Matthieu R.},
    title={Experimental aspects of secret key generation in indoor wireless environments}, 
    booktitle=IEEE_C_SPAWC, 
    year={2013},
    pages={669-673},
    doi={10.1109/SPAWC.2013.6612134}
}

@InProceedings{Evolving_Secret_Sharing,
    author={Komargodski, Ilan and Paskin-Cherniavsky, Anat},
    title={Evolving Secret Sharing: Dynamic Thresholds and Robustness},
    booktitle=SPRINGER_C_TCC,
    pages={379-393},
    year={2017},
    organization={Springer}
}

@InProceedings{Capacity_of_a_shared_secret_key,
    author = {Csisz{\'a}r, Imre and Narayan, Prakash},
    title = {Capacity of a shared secret key}, 
    booktitle = IEEE_C_ISIT, 
    year = {2010},
    pages = {2593-2596},
    doi = {10.1109/ISIT.2010.5513769}
}

@InProceedings{Secret_Sharing_Over_a_Gaussian_Broadcast_Channel_Optimal_Coding_Scheme_Design_and_Deep_Learning_Approach_at_Short_Blocklength,
    author = {Sultana, Rumia and Rana, Vidhi and Chou, R\'{e}mi A.},
    title={Secret Sharing Over a {G}aussian Broadcast Channel: Optimal Coding Scheme Design and Deep Learning Approach at Short Blocklength}, 
    booktitle = IEEE_C_ISIT, 
    year = {2023},
    pages = {1961-1966},
    doi = {10.1109/ISIT54713.2023.10206773}
}

@InProceedings{Low-Complexity_Secret_Sharing_Schemes_Using_Correlated_Random_Variables_and_Rate-Limited_Public_Communication,
    author = {Sultana, Rumia and Chou, R\'{e}mi A.},
    title = {Low-Complexity Secret Sharing Schemes Using Correlated Random Variables and Rate-Limited Public Communication},
    booktitle = IEEE_C_ISIT,
    year = {2021},
    publisher = {IEEE Press},
    url = {https://doi.org/10.1109/ISIT45174.2021.9517753},
    doi = {10.1109/ISIT45174.2021.9517753},
    pages = {970–975},
    numpages = {6},
    location = {Melbourne, Australia}
}

@InProceedings{Secret_Sharing_on_Evolving_Multi-level_Access_Structure,
    author={Dutta, Sabyasachi and Roy, Partha Sarathi and Fukushima, Kazuhide and Kiyomoto, Shinsaku and Sakurai, Kouichi},
    title={Secret Sharing on Evolving Multi-level Access Structure},
    booktitle=SPRINGER_C_ISA,
    year={2020},
    publisher={Springer International Publishing},
    pages={180-191},
    isbn={978-3-030-39303-8}
}

@InProceedings{Evolving_Secret_Sharing_Made_Short,
    title={Evolving secret sharing made short},
    author={Francati, Danilo and Venturi, Daniele},
    booktitle=SPRINGER_C_ICTACIS,
    pages={69-99},
    year={2024}
}

@phdthesis{Cuff_Thesis,
    author = {Cuff, Paul W.},
    advisor = {Cover, Thomas M.},
    title = {Communication in networks for coordinating behavior},
    year = {2009},
    isbn = {9781109444315},
    publisher = {Stanford University},
    address = {Stanford, CA, USA},
    note = {AAI3382710}
}

\end{document}